\documentclass[conference]{IEEEtran}
\usepackage{amsmath}
\usepackage{graphicx}
\newtheorem{lemma}{Lemma}
\ifCLASSINFOpdf
\else
\fi
\begin{document}
%
\title{Cell Association and Handover Management in Femtocell Networks}

\author{\IEEEauthorblockN{Hui Zhou, Donglin Hu, Shiwen Mao, Prathima Agrawal, Saketh Anuma Reddy}
\IEEEauthorblockA{Dept. of Electrical and Computer Engineering, Auburn University, Auburn, AL 36849-5201, USA \\
Email: hzz0016, dzh0003, szm0001, agrawpr, sza0041@auburn.edu}}


%


\maketitle

\begin{abstract}
Although the technology of femtocells is highly promising, many challenging problems should be addressed before fully harvesting its potential. In this paper, we investigate the problem of cell association and handover management in femtocell networks. Two extreme cases for cell association are first discussed and analyzed. Then we propose our algorithm to maximize network capacity while achieving fairness among users. Based on this algorithm, we further develop a handover algorithm to reduce the number of unnecessary handovers using Bayesian estimation. The proposed handover algorithm is demonstrated to outperform a heuristic scheme with considerable gains in our simulation study. 

\end{abstract}


%
\IEEEpeerreviewmaketitle

\section{Introduction}

Due to the nature of open space used as wireless transmission medium, wireless network capacity is largely limited by interference. Mobile users at the border of cellular networks require considerably large transmit power to overcome attenuation, which in return causes interference to other users and reduces network capacity. To address this issue, femtocells provide an effective solution by shortening transmission distance and bringing base stations (BS) closer to mobile users~\cite{Chandrasekhar08}. 

As shown in Fig.~\ref{fig:femtocell}, a femtocell is usually a small celluar network, with a femtocell base station (FBS) connected to macrocell base station (MBS) via broadband wireline. The FBS is usually deployed at the places where the strength of the signal received from MBS is weak (e.g., indoor, the edge of MBS). It is designed to offload MBS traffic and serve the approved users when they are within the coverage. Due to the reduced distance of wireless transmission, femtocell is shown effective in reducing transmit power~\cite{Hu11GC} and improving signal-to-interference-plus-noise ratio (SINR), which lead to prolonging battery life of mobile devices and enhancing network coverage as well capacity~\cite{Chandrasekhar08}.
\begin{figure}[ht]
	\centering
		\includegraphics[width=2.5in]{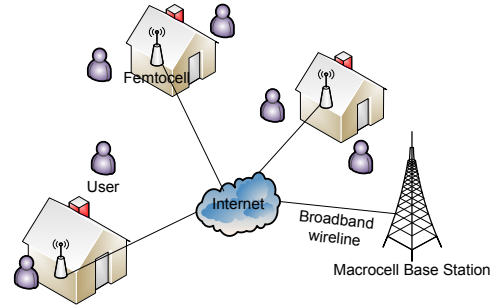}
	\caption{Femtocell networks}
	\label{fig:femtocell}
\end{figure}

Femtocells have attracted significant interest from wireless industry. Major wireless network operators in United States such as AT\&T, Sprint and Verizon, have provided femtocell service plan recently. Although the potential of femtocells is highly promising, a broad range of problems across technical issues, regulatory concerns and economic incentives should be addressed~\cite{Andrews12}. The challenging technical issues include synchronization, cell association, mobility and handover, interference mitigation, network organization, and quality of service (QoS) provisioning~\cite{Hu11IDS}.

In this paper, we investigate the problem of cell association and handover management in femtocell networks. We consider an MBS and multiple FBS's deployed in a femtocell network. The MBS and FBS's cooperatively send data to users in the network through downlink transmission. Each user is allowed to connect to either the MBS or an FBS. Open access mechanism is employed since it is shown significantly higher network capacity than closed access system~\cite{Cheung12}. In open access systems, all users have chance to connect to each FBS. The problem is to decide the assignment between BS's and users with the objective of maximizing network throughput and achieving fairness among users. Besides, when users are in motion, reducing the number of handerovers and handover delay is critical for the success of femtocell technology.
  
We provide an analysis of network downlink capacity for two extreme cases. In case I, each BS selects the best user and total network capacity is maximized. However, in case II, each user chooses the best BS to connect and fairness is achieved among users. Then we propose a new cell association algorithm to find a trade-off between two extreme cases. Based on the cell association algorithm, we further present a handover algorithm as well an access control algorithm for mobile users. 

The remainder of this paper is organized as follows. The related work is discussed in Section~\ref{sec:RelWork}. We present the system model and analyze two extreme cases in Section~\ref{sec:sysMod}. In Section~\ref{sec:PropSchm}, cell association and handover management algorithms are proposed. The proposed algorithms are evaluated in Section~\ref{sec:PerfEva}. Section~\ref{sec:Conc} concludes this paper.

\section{Related Work}\label{sec:RelWork}
Femtocells have received considerable interest from both industry and academia. Technical challenges, regulatory requirements and economic concerns in femtocell networks are comprehensively discussed in~\cite{Chandrasekhar08} and~\cite{Andrews12}. Since FBS's are distributedly deployed and are able to spatially reuse the same spectrum belonging to the MBS, many research efforts were made on interference mitigation by assigning users to the proper BS. In~\cite{Roche10}, the disadvantages of open and closed access mechanisms were discussed and a hybrid access control scheme was introduced. In~\cite{Cheung12}, a stochastic geometric model was introduced to derive the success probability for MBS's and FBS's under open and closed access schemes. 
A learning-based cell selection method for an open access femtocell network was proposed in~\cite{Dhahri12}.  

The main objective of handover algorithm is to determine an optimal connection while minimizing handover latency and reducing unnecessary handovers. In~\cite{Moon10}, an efficient handover algorithm was proposed by considering the optimal combination of received signal strengths from a serving MBS and a target FBS. In~\cite{Jeong11}, a new handover decision algorithm based on mobility pattern and location prediction was developed to reduce the number of unnecessary handovers due to temporary femtocell visitors. Both signal strength and velocity were considered for the proposed handover algorithm in~\cite{Wu09}. A hybrid access scheme and a femtocell-initiated handover procedure with adaptive threshold were studied in~\cite{Fan10}.
  
\section{System Model}\label{sec:sysMod}
We consider a femtocell network with an MBS (indexed $0$) and $M$ FBS's (indexed from $1$ to $M$), which is illustrated in Fig.~\ref{fig:femtocell}. The MBS and $M$ FBS's are connected to the Internet via broadband wireline connection, where $N$ mobile users are randomly located inside the macrocell coverage area. We assume the MBS and $M$ FBS's are well synchronized and they occupy the same spectrum at the same time to send data to mobile users.

Let $P_0$ be the MBS transmit power and $h_{0,k}$ be the channel gain between the MBS and $k$-th user. Likewise, $P_i$ and $h_{i,k}$ where $i\ge 1$ denote the transmit power of the $i$-th FBS as well as the channel gain between the $i$-th FBS and $k$-th user. We assume an additional white Gaussian noise (AWGN) at mobile users with power density $\sigma^2$. The capacity at the $k$-th user from its serving MBS is given by:
\begin{eqnarray}\label{eq:C_MBS}
C_k=\frac{B}{N_0}\log_2\left(1+\frac{|h_{0,k}|^2P_0}{\sigma^2+I_{0,k}}\right)
\end{eqnarray}
where $B$ is the network bandwidth, $N_0$ is the number of MBS users, and $I_{0,k}=\sum_{i=1}^M|h_{i,k}|^2P_i$ is the interference from FBS's. We assume the bandwidth is equally allocated to all served users. The capacity at the $j$-th user from the $i$-th FBS is given by:
\begin{eqnarray}\label{eq:C_FBS}
C_j=\frac{B}{N_i}\log_2\left(1+\frac{|h_{i,j}|^2P_i}{\sigma^2+I_{i,j}}\right)
\end{eqnarray}
where $N_i$ is the number of users served by the $i$-th FBS and $I_{i,j}=\sum_{l=0,l\neq i}^M|h_{l,j}|^2P_l$ is the interference from the MBS and other FBS's.

By combining (\ref{eq:C_MBS}) and (\ref{eq:C_FBS}), we have the following equation for the capacity at the $j$-th user from the $i$-th BS:
\begin{eqnarray}\label{eq:C_User}
C_j&=&\frac{B}{N_i}\log_2\left(1+\frac{|h_{i,j}|^2P_i}{\sigma^2+I_j-|h_{i,j}|^2P_i}\right) \nonumber \\
&=&\frac{B}{N_i}\log_2\left(\frac{\sigma^2+I_j}{\sigma^2+I_j-|h_{i,j}|^2P_i}\right) \nonumber \\
&=&\frac{B}{N_i}\log_2\left(\frac{1}{1-\eta_{i,j}}\right)
\end{eqnarray}
where $I_j=\sum_{i=0}^M|h_{i,j}|^2P_i$ is the sum of received power from its serving BS and interference from other BS's, and $\eta_{i,j}=|h_{i,j}|^2P_i/(\sigma^2+I_j)$ is SINR, which is the percentage of desired power in $I_j$. Note that $I_j$ does not depend on which BS the user is connected to, and it is a constant for any BS. 

\subsection{Case I: Network Capacity}
Initially, our objective is to maximize the total network capacity. By denoting $U_i$ as the set of users connected to the $i$-th BS, we have $N_i=|U_i|$. Then, by using (\ref{eq:C_User}), the objective function can be expressed as:
\begin{eqnarray}\label{eq:C_tot}
\mbox{Maximize: } C_{tot}=\sum_{i=0}^M \frac{B}{N_i}\sum_{j\in U_i}\log_2\left(\frac{1}{1-\eta_{i,j}}\right).
\end{eqnarray}
The optimal solution to the problem above is that each BS chooses one user with the highest SINR to connect. 
This solution is able to achieve the highest network throughput by assigning only one best user to each BS. The rest of users are not allowed to access the network. This solution is unfair and inefficient because only a small portion of users are served. With this scheme, this system can only accommodate at most $M+1$ (the number of BS's) users. 


\subsection{Case II: User Fairness}
To achieve fairness among users, we divide the bandwidth equally and allocate them to all users connected to the same BS. Then, a straightforward heuristic solution is proposed that each user $j$ chooses a BS with the highest SINR to connect. However, this approach may incur the QoS problems, especially when all users choose the same BS to connect. Each user is assigned with a very small bandwidth which leads to extremely low capacity. for each user. On the other hand, the users with low SINR from any of BS's may jeopardize the total network throughput~\cite{Hu10JSAC}. Obviously, blocking these users can improve the total network capacity. To guarantee the minimum QoS requirements of each user and maximize the total network capacity, only users with SINR above $\lambda_1$ are allowed to access network and each BS is able to serve at most $N_{max}$ users. 

\section{Proposed Scheme}\label{sec:PropSchm}
From our previous analysis in case I, we find that the total network capacity is maximized if each BS chooses only one user with the highest SINR. However, this scheme is not fair for the other users because they do not have chance to be served. Although the scheme in case II is fair, the network capacity is very low. Therefore, we want to find a trade-off between network performance and fairness.

\subsection{User Classification}
Before introducing our scheme, we adopt $q$ thresholds $\lambda_i$'s to divide SINR's into $q+1$ levels:
\begin{eqnarray}\label{eq:LDef}
L_{i,j}=\left\{\begin{array}{l l}
 0, &  \eta_{i,j}<\lambda_1 \\
 l, &  \lambda_{l}\le\eta_{i,j}<\lambda_{l+1},l\in \{1,\cdots,q-1\} \\
 q,& \eta_{i,j}\ge\lambda_{l+1}
\end{array}\right. 
\end{eqnarray} 
According to $L_{i,j}$, the users are divided into $q+1$ groups. Our idea is to group these users and let them connect to the same BS. Since the SINR values of the users in the same group are very close, we can replace individual SINR with the average value. Then, the objective function in (\ref{eq:C_tot}) can be rewritten as:
\begin{eqnarray}\label{eq:C_tot_lev}
C_{tot}&=&\sum_{i=0}^M \frac{B}{N_i}\sum_{l=0}^q\sum_{j\in U_{i,l}}\log_2\left(\frac{1}{1-\eta_{i,j}}\right) \nonumber \\
&\approx& \sum_{i=0}^M B\sum_{l=0}^q\frac{N_{i,l}}{N_i}\log_2\left(\frac{1}{1-\overline{\eta}_{i,l}}\right)
\end{eqnarray}
where $U_{i,l}=\{j|j\in U_i,L_{i,j}=l\}$, $N_{i,l}=|U_{i,l}|$ is the number of users in $U_{i,l}$ and $\overline{\eta}_{i,l}$ is the average value of SINR in $U_{i,l}$. If we denote $C_{NC}$ and $C_{UC}$ as the total network capacity achieved by case I and proposed scheme, respectively, we have the following lemma.

\begin{lemma}
$C_{NC}$ is a upper bound of $C_{UC}$.
\end{lemma}
\begin{proof}
Although it is obvious that $C_{NC}$ is greater than $C_{UC}$, we provide detailed mathematic analysis proof. Due to the fact that arithmetic mean is always equal to or greater than geometric mean and the equality holds if all numbers are equal, we have the following inequality:
\begin{eqnarray}\label{eq:Ineq}
\prod_{j\in U_{i,l}}(1-\eta_{i,j})\le (1-\overline{\eta}_i)^{N_{i,l}}
\end{eqnarray}
By taking inversion and logarithm to the both sides of the inequality above, we have: 
\begin{eqnarray}\label{eq:C_UB}
\sum_{j\in U_{i,l}}\log_2 \left(\frac{1}{1-\eta_{i,j}}\right)\ge N_{i,l}\log_2 \left(\frac{1}{1-\overline{\eta}_{i,l}}\right)
\end{eqnarray}
Therefore, the lemma holds.
\end{proof}  
 
Although $C_{NC}$ is a upper bound of $C_{UC}$, their values are very close due to the fact that the SINR values of the users from the same group fall within the same range and are close to each other. Therefore we can adopt (\ref{eq:C_tot_lev}) as objective function instead of (\ref{eq:C_tot}). Obviously, to maximize $C_{UC}$, the number of the users at the $q+1$ level $N_{i,q}$ should be equal to $N_i$ since $\overline{\eta}_{i,l}>\overline{\eta}_{i,l-1}$ for any $0<l\le q$. Then, we can merge users that do not connect to BS's into one group and divide users into three groups with $q=2$. In the first group, the SINR's of these users are too low to allow them to connect to any of the BS's. Comparatively, each user in the third group is connected to one of the BS's. The rest of users in the second group are candidate that will be admitted to the BS's when BS's have available resource to allocate. Then, we rewrite $C_{UC}$ as:
\begin{eqnarray}\label{eq:C_UC}
C_{UC}=\max \sum_{i=0}^M B\log_2\left(\frac{1}{1-\overline{\eta}_{i,2}}\right)
\end{eqnarray}
Note that $C_{UC}$ only depends on the average SINR of users in the third class.


Until now, we assume all BS's adopt the same $\lambda$ to classify the users. By considering the distribution of users and BS, we denote $\lambda_l^i$ as the threshold adopted by BS $i$ to push users onto different BS's. To simplify the analysis, we assume $\lambda_1^i=\lambda_1$ for all BS's. 
An algorithm with objective of finding a trade-off between network capacity and user fairness is presented in table~\ref{tab:CellAssoc}. In step $2$, the users with all SINRs below the threshold $\lambda_1$ are removed from user set. From step $4$ to step $16$, each BS finds its candidate user set according to thresholds $\lambda_2^i$. Then, each user in the candidate user set can make its candidate BS list accordingly. The user on the candidate user list is allowed to choose the best BS from its candidate BS list. Once the BS is assigned with $N_{max}$ user, it is removed and not available for the rest of users. In steps $17-19$, the BS's adjust its $\lambda_2^i$ according to the number of users that have already been assigned. If the number of assigned users is small, the threshold $\lambda_2^i$ is reduced with large step-size $\Delta$. Once the execution of algorithm is completed, the BS with larger $\lambda_2^i$ usually has higher network capacity due to the higher average SINR.

\begin{table}[ht]
\begin{center}
\caption{Cell Association Algorithm}
\begin{tabular}{ll}
\hline
1: & Initialize $\eta_{i,j}$, $\lambda_1^i$, $\lambda_2^i$, $\mathcal{F}$, $\mathcal{U}$ and $\Delta$ \\
2: & Remove users $\{j|\eta_{i,j} < \lambda_1, \forall i\}$ from $\mathcal{U}$ \\
3: & While $\mathcal{F}$ is not empty and $\mathcal{U}$ is not empty\\
4: & $\;\;$ Find candidate user set: $V_i=\{j|j\in \mathcal{U}, \eta_{i,j} \ge \lambda_2^i\}$ \\
5: & $\;\;$ If $\cup_{i=0}^M V_i$ is empty \\
6: & $\;\;\;\;$ The algorithm is terminated \\
7: & $\;\;$ End if \\
8: & $\;\;$ Find candidate BS set: $W_j=\{i|i\in \mathcal{F}, j\in V_i, \forall j\in \cup_{i=0}^M V_i\}$ \\
9: & $\;\;$ For $j \in \cup_{i=0}^M V_i$ \\
10: & $\;\;\;\;$ Find the $i^\ast$ BS: $i^\ast=\arg\max_{i\in W_j}\eta_{i,j}$  \\
11: & $\;\;\;\;$ Add user $j$ to $U_{i^\ast}$ and set $a_{i^\ast,j}=1$ \\
12: & $\;\;\;\;$ Remove user $j$ from $\mathcal{U}$ \\
13: & $\;\;\;\;$ If $N_{i^\ast}=N_{max}$\\
14: & $\;\;\;\;\;\;$ Remove the BS $i^\ast$ from $\mathcal{F}$ and all $W_j$'s \\
15: & $\;\;\;\;$ End if\\
16: & $\;\;$ End for\\
17: & $\;\;$ For $i \in \mathcal{F}$ \\
18: & $\;\;\;\;$ Adjust $\lambda_2^i=\max\{\lambda_2^i-|N_{max}-N_i|\Delta,\lambda_1^i\}$ \\
19: & $\;\;$ End for\\
20: & End while\\
\hline
\end{tabular}
\label{tab:CellAssoc}
\end{center}
\end{table}  


\subsection{Handover Algorithm}
Previously, we discuss the cell association when all users are not in motion. Now, we consider user mobility in this section. Since all users may travel from one cell to another, an efficient handover algorithm is essential to handle this issue. Before presenting our algorithm, we have to introduce several notations used in the algorithm. First, we denote $\Omega_i$ as the coverage of the $i$-th BS, in which the received power from the $i$-th BS is dominant among all received power from all BS's. $\pi_i$ is denoted as the probability that user is in the coverage of the $i$-th BS. Since the coverage of femtocell is very small, the probability of user in the coverage of MBS, $\pi_0$, is much higher than the other $\pi_i$'s ($i\neq 0$). In addition, we define a conditional probability $\epsilon_i=\Pr(\eta_{i',j}>\eta_{i,j}|Loc(j)\in\Omega_i)$ as the probability that SINR from the $i'$-th BS is greater than that from the $i$-th BS conditioned on user $j$ is in the coverage of the $i$-th BS where $Loc(j)$ is the location of the user $j$. Then, we collect SINR information from all BS's for $T$ times and count the times that SINR from the BS $i$ is less than those from the other BS's, denoted by $n_i$. Thus, SINR's are compared for totally $M\times T$ times. With comparison results of SINR, denoted by $\Theta$, we can adopt Bayesian estimation to estimate the posterior probability that user $j$ is in the coverage of BS $i$ as:
\begin{eqnarray}\label{eq:Prob_BS}
Q_i&=&\Pr(Loc(j)\in \Omega_i|\Theta) \nonumber \\
&=&\frac{\Pr(\Theta|Loc(j)\in \Omega_i)\pi_i}{\sum_{i=0}^M\Pr(\Theta|Loc(j)\in \Omega_i)\pi_i} \nonumber \\
&=&\frac{\epsilon_i^{n_i}(1-\epsilon_i)^{MT-n_i}\pi_i}{\sum_{i=0}^M\epsilon_i^{n_i}(1-\epsilon_i)^{MT-n_i}\pi_i}.
\end{eqnarray}

The proposed handover algorithm is presented in table~\ref{tab:Handover}. In steps $3-4$, the connection between user and BS is terminated because the SINR requirement at user $j$ cannot be satisfied by BS $i$. Note that $\overline{\eta}_{i,j}$ is the average SINR over $T$ times. In steps $5-12$, we compute the posterior probability $Q_i$ and find available BS's with $Q_i$ above a predefined threshold $\Gamma$. Among all available BS's, user $j$ chooses a best $BS$ to connect.
 
\begin{table}[ht]
\begin{center}
\caption{Handover Algorithm for User $j$ connecting to BS $i$}
\begin{tabular}{ll}
\hline
1: & While TRUE \\
2: & $\;\;$ Collect SINR $\eta_{i,j}^t$ ($t=1,\cdots,T$) from all BS's\\
3: & $\;\;$ If $\overline{\eta}_{i,j}<\lambda_1$\\
4: & $\;\;\;\;$ Break the connection between $i$ and $j$ \\
5: & $\;\;$ Else if $\overline{\eta}_{i,j}<\lambda_2^i$   \\
6: & $\;\;\;\;$ Compute posterior probability $Q_i$ \\
7: & $\;\;\;\;$ $W_j=\{i|N_i<N_{max} \mbox{ and } Q_i\ge \Gamma \mbox{ and } \overline{\eta}_{i,j}\ge \lambda_2^i\}$\\
8: & $\;\;\;\;$ If $W_j$ is not empty \\
9: & $\;\;\;\;\;\;$ Find the BS $i^\ast=\arg\max_{i\in W_j} Q_i$\\
10: & $\;\;\;\;\;\;$ Break the connection with BS $i$ and connect to the BS $i^\ast$ \\
11: & $\;\;\;\;$ End if\\
12: & $\;\;$ End if\\
13: & End while\\
\hline
\end{tabular}
\label{tab:Handover}
\end{center}
\end{table}  

\subsection{Admission Algorithm}
According to the handover algorithm proposed before, handover procedure does not always succeed due to low SINR or busy BS. The connect of served users may be dropped. Therefore, the problem of letting users admitted or readmitted to the femtocell network should be addressed. To solve this problem, we present an admission algorithm in table~\ref{tab:Admin}. It is similar to the proposed handover algorithm expect that the users do not need to break the connect with previous BS.
\begin{table}[b]
\begin{center}
\caption{Admission Algorithm for Inactive User $j$ }
\begin{tabular}{ll}
\hline
1: & While TRUE \\
2: & $\;\;$ Collect SINR $\eta_{i,j}^t$ ($t=1,\cdots,T$) from all BS's\\
3: & $\;\;$ Compute posterior probability $Q_i$ \\
4: & $\;\;$ $W_j=\{i|N_i<N_{max} \mbox{ and } Q_i\ge \Gamma \mbox{ and } \overline{\eta}_{i,j}\ge \lambda_2^i\}$\\
5: & $\;\;$ If $W_j$ is not empty \\
6: & $\;\;\;\;$ Find the BS $i^\ast=\arg\max_{i\in W_j} Q_i$\\
7: & $\;\;\;\;$ Connect to the $i^\ast$ BS \\
8: & $\;\;$ End if\\
9: & End while\\
\hline
\end{tabular}
\label{tab:Admin}
\end{center}
\end{table}  

\section{Performance Evaluation}\label{sec:PerfEva}
We evaluate the performance of the proposed cell association and handover algorithms using MATLAB. Each point in the following figures is the average of $10$ simulation runs. The $95\%$ confidence intervals are plotted for each point. We adopt the similar channel models used in~\cite{Moon10}.The values of channel gain from the BS's can be expressed as: 
\begin{eqnarray}\label{eq:RSS} 
&&10\log_{10}\left(|h_{i,j}(t)|^2\right) \nonumber\\
&=&-PL_i(t)-u_i(t)\nonumber\\
&=&-d_i-c_i\log_{10}\left[d_{i,j}(t)\right]-u_i(t) \nonumber
\end{eqnarray}
where $d_i$ and $c_i$ are two constants for path loss model $PL_i$, $d_{i,j}(t)$ is the distance from user $i$ to BS $j$ at time $t$ and $u_i(t)$ represents shadowing effect which is normally distributed with mean zero and variance $\delta_i$. The simulation parameters are listed in table~\ref{tb:Parameter}.
\begin{table} 
\caption{Simulation Parameters}
\begin{tabular}{l|l}
\hline
{\em Symbol} & {\em Definition} \\
\hline
$M=9$ & The number of femtocells \\
$B=10 \mbox{ MHz}$ & Total network bandwidth \\
$P_0=43 \mbox{ dBm}$ & Transmit power of the MBS \\
$P_i=31.5 \mbox{ dBm}$ & Transmit power of the $i$-th FBS \\
$PL_0=28+35\log_{10}(d)$ & Path loss model for MBS\\
$PL_i=38.5+20\log_{10}(d)$ & Path loss model for FBS \\
$\delta_0,\delta_i=6 \mbox{ dB}$ & Shadowing effects for MBS and FBS\\
$N_{max}=10$ & Maximum number of users per BS\\
\hline
\end{tabular}
\label{tb:Parameter}
\end{table}

For the proposed cell association algorithm, we compare it with the two straightforward schemes discussed in Section~\ref{sec:sysMod}:
\begin{itemize}
	\item Scheme $1$ based on maximizing network capacity: each BS chooses the user with the highest SINR to connect. 
	\item Scheme $2$ based on fairness among users: each user connects to the BS from which it can receive the highest SINR. 
\end{itemize}


In Fig.~\ref{fig:NetCapVSNumUser}, we examine the impact of the number of users on the total network capacity. We increase $N$ from $20$ to $100$ with step-size $20$, and plot the total network capacity. We find that the total network capacity increases with the number of users because the probability that BS's choose a user with better SINR to connect becomes higher as the number of users grows larger. As expected, scheme $1$ achieves the highest network capacity, while network capacity in scheme $2$ is the lowest since the users with poor SINR jeopardize the total network performance by occupying a portion of network bandwidth. The network capacity of the proposed scheme is almost as twice as scheme $2$ when the number of users is close to $100$. Although the network capacity of the proposed scheme is about one half of that of scheme $1$, note that the number of served users in the proposed scheme achieves as $N_{max}=10$ times as that in scheme $1$.

\begin{figure}[ht]
	\centering
		\includegraphics[width=2.5in]{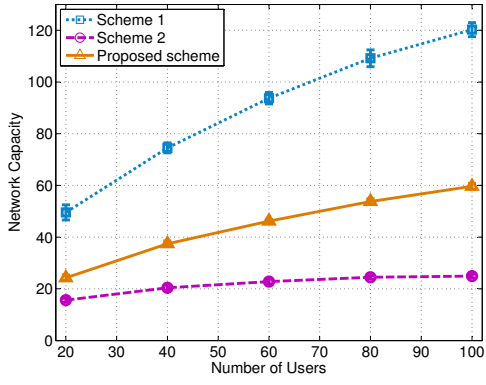}
	\caption{Total network capacity vs. number of users}
	\label{fig:NetCapVSNumUser}
\end{figure}

Then, we adopt Raj Jain's fairness index to investigate the impact of number of users on fairness among users. Jain's equation is given by:
\begin{eqnarray}
\mathcal{J}(C_1,C_2,\cdots,C_N)=\frac{(\sum_{j=1}^NC_j)^2}{N\times\sum_{j=1}^NC_j^2}
\end{eqnarray}
where $C_j$ is the network throughput for user $j$. The value of the index ranges from $1/N$ (worst case) to $1$ (best case). It can be seen from Fig.~\ref{fig:NumHOVSNumUser} that fairness indexes decrease with the number of users. Scheme $1$ has the lowest fairness index since it only serves at most $M+1$ users. However, scheme $2$ obtains the highest fairness index among three schemes because every user in scheme $2$ has chance to connect to BS. Although the proposed scheme is not as fair as scheme $2$, their fairness indexes are very close when the number of users is around $20$. 
\begin{figure}[ht]
	\centering
		\includegraphics[width=2.5in]{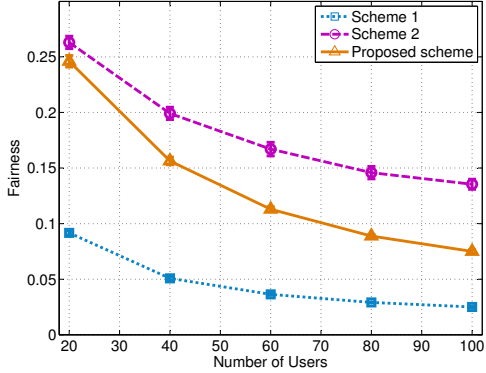}
	\caption{Fairness vs. number of users}
	\label{fig:FairnessVSNumUser}
\end{figure}

Next, we evaluate the performance of our proposed handover algorithm by applying random walk model to each user. Both velocity and direction of mobile users are uniformly distributed within $[0, 8.3]$ m/s and $[0, 2\pi]$, respectively. Since we do not find any similar schemes in the literature, we compare the proposed scheme with a heuristic scheme: Once the average SINR $\overline{\eta}_{i,j}$ falls below threshold $\lambda_2^i$, user $j$ will choose a best available BS to connect.


In Fig.~\ref{fig:NumHOVSNumUser}, we show the impact of number of users on the average number of handovers. We increase $N$ from $20$ to $100$ with step-size $20$. We find that when the number of users is less than $60$, the average number of handovers in the heuristic scheme grows larger with the number of users. It is due to the fact that the more users, the more frequently handovers take place. Once the number of users gets beyond $60$, the average number of handovers decreases because the probability of  finding available BS gets smaller. However, the average number of handovers in our proposed scheme is much lower than the that of heuristic scheme and decreases slowly with the number of users.     

%
\begin{figure}[ht]
	\centering
		\includegraphics[width=2.5in]{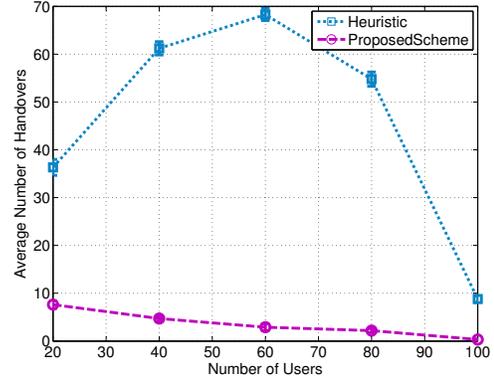}
	\caption{Average number of handovers vs. number of users}
	\label{fig:NumHOVSNumUser}
\end{figure}

Finally, we examine the impact of maximum allowed number of user per BS on the average number of handovers in Fig.~\ref{fig:NumHOVSNmax}. We increase $N_{max}$ from $2$ to $10$ with step-size $2$. When $N_{max}$ is below $4$, the average number of handovers is close to $0$ for both heuristic and proposed scheme because all BS are busy and users are not allowed to connect to the new BS. Beyond this critical point, the average number of handovers in proposed scheme is significantly reduced compared with heuristic scheme.  

\begin{figure}[ht]
	\centering
		\includegraphics[width=2.5in]{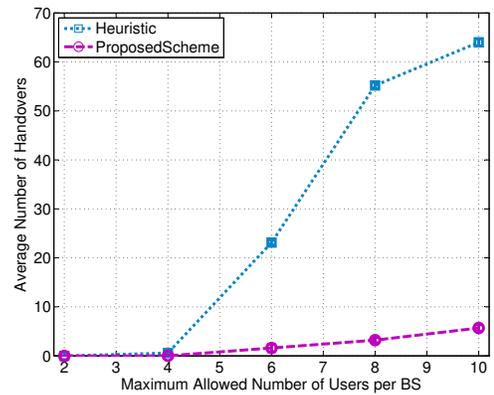}
	\caption{Average number of handovers vs. $N_{max}$}
	\label{fig:NumHOVSNmax}
\end{figure}

\section{Conclusion}\label{sec:Conc}
In this paper, we investigated the problem of cell association and handover management in femtocell networks consisting of an MBS and multiple FBS's. We first proposed a cell association algorithm with the objective of seeking a trade-off point between network capacity and fairness. Based on this algorithm, we presented a handover algorithm for mobile users. Both cell association and handover algorithm were evaluated with simulations. The handover algorithm was shown to outperform a heuristic scheme with considerable gains.


\section*{Acknowledgment}
This work is supported in part by the U.S. National Science Foundation
(NSF) under Grants CNS-1145446 and CNS-1247955, and through the NSF Wireless Internet Center for Advanced Technology at Auburn University.

\bibliographystyle{IEEEtran}
\bibliography{handoff_femto}

\end{document}